\documentclass[ 
reprint, 
superscriptaddress,
amsmath,
amssymb, 
floatfix,
longbibliography,
aps,
prl,
]{revtex4-1}

\usepackage{graphicx}
\usepackage{xcolor}
\usepackage[colorlinks=true,citecolor=blue,linkcolor=magenta]{hyperref}
\usepackage{bm}
\usepackage{amsthm}

\usepackage{mathtools}

\DeclareMathOperator{\Tr}{Tr}

\DeclareMathOperator{\var}{var}
\DeclareMathOperator{\med}{med}

\DeclareMathOperator{\QFT}{QFT}

\newtheorem{theorem}{Theorem}
\newtheorem{proposition}[theorem]{Proposition}
\newtheorem{protocol}[theorem]{Protocol}
\newtheorem{lemma}[theorem]{Lemma}
\newtheorem{corollary}[theorem]{Corollary}

\newcommand{\GHZ}{{\rm GHZ}}

\begin{document}

\title{Direct Fidelity Estimation for Generic Quantum States}
\author{Christopher Vairogs}
\affiliation{Theoretical Division, Los Alamos National Laboratory, Los Alamos, New Mexico 87545, USA}
\affiliation{Department of Physics, University of Illinois at Urbana-Champaign, Urbana, Illinois 61801, USA}
\author{Bin Yan}
\affiliation{Theoretical Division, Los Alamos National Laboratory, Los Alamos, New Mexico 87545, USA}

\date{\today}

\begin{abstract}
Verifying the proper preparation of quantum states is essential in modern quantum information science. Various protocols have been developed to estimate the fidelity of quantum states produced by different parties. Direct fidelity estimation is a leading approach, as it typically requires a number of measurements that scale linearly with the Hilbert space dimension, making it far more efficient than full state tomography. In this article, we introduce a novel fidelity estimation protocol for \emph{generic} quantum states, with an overall computational cost that scales only as the square root of the Hilbert space dimension. Furthermore, our protocol significantly reduces the number of required measurements and the communication cost between parties to \emph{finite}. This protocol leverages the quantum amplitude estimation algorithm in conjunction with classical shadow tomography to achieve these improvements.
\end{abstract}

\maketitle

As quantum devices become increasingly sophisticated, it is crucial to verify that their subroutines prepare quantum states as intended. Specifically, efficiently and accurately certifying that a quantum state has high fidelity with a target state has garnered much attention. For instance, imagine the situation that, as illustrated in Fig.~\ref{fig:illustration}, agent Alice needs to certify a quantum device developed by Bob at a distant location. The task is to directly estimate the fidelity between quantum states generated by their respective devices using only local operations and classical communications.

Direct fidelity estimation (DFE) is one of the most well-known proposals for this task~\cite{Flammia2011, daSilva2011, Huang2024Certifying}. This approach relies on measuring the expectation values of Pauli strings selected via importance sampling. More recently, Classical Shadow Tomography (CST) and its variants~\cite{Huang2020, Helsen2023, Struchalin2021, Elben2022, Zhao2021} have emerged as promising tools for fidelity estimation. The CST protocol employs repeated randomized measurements to construct a representation of an unknown state $\rho$ on a classical computer, from which one can compute its fidelity with a target state $|\psi\rangle$ to arbitrary accuracy. For certain classes of quantum states with classical representations, these methods can achieve highly accurate fidelity estimation with a finite number of measurements independent of the system size. 

However, using states with classical representations for device certification poses a significant risk: Bob may efficiently fake the data to be shared with Alice and spoof her certification process. For generic states without classical representations, these approaches, as well as other proposals~\cite{Seshadri2024, Seshadri2024-2}, require a number of copies of the quantum states (measurements) that generally scale as $\mathcal{O}(d)$, where $d$ is the dimension of the system's Hilbert space. Nevertheless, such scaling significantly outperforms full state tomography by a factor of $d$.

Here, we present an approach that further reduces the computational cost to scale as the square root of the Hilbert space dimension, setting a new benchmark for generic state fidelity estimation. 

In our setup, we assume that one party Alice has access to an arbitrary pure state $|\psi\rangle$ and another party Bob has access to an arbitrary state $\rho$ to be certified. Our approach provides a method for them to construct an estimator for the true fidelity $F \equiv \langle \psi|\rho|\psi\rangle$. The protocol involves two essential ingredients: Bob employs CST on copies of his state and sends the obtained classical data to Alice; With the received information, Alice then performs the quantum amplitude estimation (QAE) algorithm~\cite{Brassard2002} on copies of her state and obtains a fidelity estimate. 

The main result of this work is that Bob only needs to perform a finite number of measurements independent of the system size, while Alice requires $\mathcal{O}(\sqrt{d})$ iterations of QAE. This represents an overall quadratic speedup over the state-of-the art. Moreover, our protocol reduces the number of measurements and the involved communication cost between Alice and Bob to a \emph{finite} value, providing significant advantages compared to the $\mathcal{O}(d)$ scaling in conventional approaches. In the following, we will briefly introduce CST and QAE, and present our fidelity estimation protocol with proved performance guarantee and numerical simulations.

\begin{figure}[b!]
    \centering
    \includegraphics[width=\columnwidth]{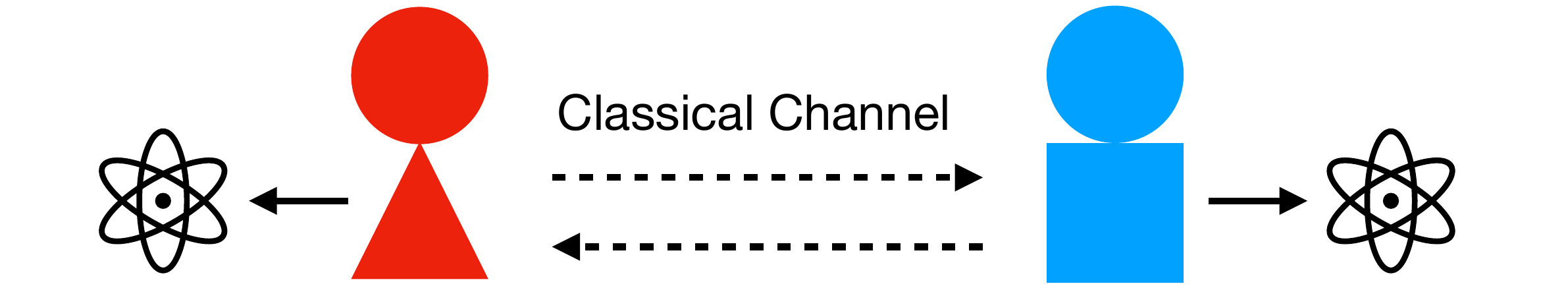}
    \caption{Alice and Bob have local access to their respective quantum states and can share information through a classical channel. The task is to estimate the fidelity between their states with the provided resources.}
    \label{fig:illustration}
\end{figure}

\vspace{8pt}
\emph{Classical Shadow Tomography} allows us to predict many properties of a quantum state with only a few measurements. Consider an $n$-qubit Hilbert space $\mathcal{H}^n$ and let $\mathcal{D}(\mathcal{H}^n)$ denote the set of density operators over $\mathcal{H}^n$. To perform CST on an unknown state $\rho \in \mathcal{D}(\mathcal{H}^n)$, one first samples a unitary $\hat{U}$ (here the hat symbol is used to indicate a random variable) from some ensemble $\mathcal{U}$ of unitaries over $\mathcal{H}^n$ and evolves $\rho$ via $\hat{U}$. The rotated state $\hat{U}\rho\hat{U}^\dagger$ is then measured in the computational basis $\{|\hat{b}\rangle: b\in \{0,1\}^n\}$, yielding an outcome $|\hat{b}\rangle$.

A key ingredient of CST is the map $\mathcal{M}: \mathcal{D}(\mathcal{H}^n) \to \mathcal{D}(\mathcal{H}^n)$ defined for an arbitrary state $\sigma \in \mathcal{D}(\mathcal{H}^n)$ as
\begin{equation}
    \mathcal{M}(\sigma) \equiv \sum_{b\in \{0,1\}^n} \mathbb{E}_{\hat{U}\sim \mathcal{U}}[\hat{U}^\dagger|b\rangle\langle b| \hat{U} \langle b|\hat{U}\sigma\hat{U}^\dagger|b\rangle].
\end{equation}
That is, $\mathcal{M}(\rho)$ is simply the average of the states $\hat{U}^\dagger |\hat{b}\rangle\langle \hat{b}| \hat{U}$ resulting from the CST procedure over the distribution of the unitaries of $\mathcal{U}$ and the distribution of the measurement outcomes. If $\mathcal{U}$ is chosen reasonably (more precisely, it is tomographically complete), then the map $\mathcal{M}$ must have a unique inverse $\mathcal{M}^{-1}$. Given an apt choice of $\mathcal{U}$, the operator $\hat{\rho} \equiv \mathcal{M}^{-1}(\hat{U}^\dagger|\hat{b}\rangle\langle \hat{b}|\hat{U})$ can be computed easily and has a classical description. We thus refer to $\hat{\rho}$ as a \textit{classical snapshot} of $\rho$. The significance of the averaging map $\mathcal{M}$ comes from the fact that
\begin{equation}\label{eq:avg-of-inverse}
    \mathbb{E}_{\hat{U}, \hat{b}}[\hat{\rho}] = \rho,
\end{equation}
which follows simply from the the definitions of $\hat{\rho}$ and $\mathcal{M}$ together with the linearity of $\mathcal{M}^{-1}$. Therefore, we can use a number of different copies of the classical snapshots $\hat{\rho}$ to predict properties of the true state $\rho$. 

The accuracy of this procedure depends on the choice of the random unitary ensemble $\mathcal{U}$ and the particular observable to be evaluated. For our purpose, the target quantity is the fidelity between $\rho$ and a known pure state $|\psi\rangle$. In this case, we can choose $\mathcal{U}$ as the uniformly weighted Clifford group and use only a finite number of sampled $\hat{\rho}$ to approximate the fidelity $\langle \psi|\rho|\psi\rangle$ to arbitrary accuracy. The corresponding classical snapshot can be evaluated as~\cite{Huang2020}
\begin{equation}\label{eq:Clifford-inverse}
    \hat{\rho} = (2^n + 1) \hat{U}^\dagger |\hat{b}\rangle\langle \hat{b}|\hat{U} - I_n.    
\end{equation}

It is worth stressing that, essential to this promise is the assumption that $|\psi\rangle$ has a classical description,  for if it does not, it cannot be processed efficiently with the snapshots on a classical computer. We address this limitation by employing the following quantum amplitude estimation algorithm. 

\vspace{8pt}
\emph{Quantum Amplitude Estimation.}--- The QAE algorithm is built off of Grover's algorithm~\cite{Grover1997, Nielsen00} and quantum phase estimation~\cite{Kitaev1995, Nielsen00} and allows one to more efficiently estimate the probability of certain pre-specified ``good'' outcomes of a computational basis measurement. In the basic setting of QAE, we are given a function $\chi: \{0,1\}^n \to \{0,1\}$ that labels bit strings as \textit{good} and \textit{bad}, with a bit string $z\in \{0,1\}^n$ being good if it satisfies $\chi(z) = 1$ and bad otherwise. Suppose one has access to a unitary $\mathcal{A}$ that produces state
\begin{equation}\label{eq:superposition}
    \mathcal{A}|0\rangle^{\otimes n} = \sum_{z=0}^{2^n-1} \alpha_z |z\rangle,
\end{equation}
for some complex coefficients $\alpha_z$. The goal of QAE is to estimate the probability $a$ that a computational basis measurement on $\mathcal{A}|0\rangle^{\otimes n}$ will generate a good bit string. 

Several modified and improved versions of the QAE algorithm have been introduced in the past decades~\cite{Grinko2021, Suzuki2020, Wie2019, Aaronson2020}. For the sake of simplicity, we will proceed with our analysis using the classic version~\cite{Nielsen00}. In this approach, one initializes the $n$ qubits and another register of $m$ ancillary qubits all in the $|0\rangle$ state. On the state register, apply the unitary $\mathcal{A}$ to generate state~\eqref{eq:superposition} and apply a Hadamard gate to every qubit of the ancillary register. Let us define the $n$-qubit oracles 
\begin{align}
    S_0 &\equiv I^{\otimes n} - 2 (|0\rangle\langle 0|)^{\otimes n}, \\
    S_\chi &\equiv I^{\otimes n} - \sum_{\chi(z) = 1} 2|z\rangle\langle z|,
\end{align}
through which we define the unitary $\mathcal{Q} \equiv -\mathcal{A} \mathcal{S}_0 \mathcal{A}^\dagger \mathcal{S}_\chi$. For $1\leq l \leq m$, we introduce the controlled operators
\begin{align}\label{eq:controlled-op}
 U_l \equiv |0\rangle_l\langle 0|_l \otimes I^{\otimes n} + |1\rangle_l\langle 
 1|_l \otimes \mathcal{Q}^{2^{l-1}},
\end{align}
where $|i\rangle_l\langle i|_l$ denotes the projector $|i\rangle\langle i|$ acting on the $l$-th qubit of the ancillary register for $i = 0,1$. The next step of the QAE algorithm is to successively apply the operators $U_1, \dots, U_{m}$ to the circuit, followed by applying the inverse quantum Fourier transform $\QFT_m^\dagger$ on the ancillary register. The final step of the algorithm is to measure the ancillary qubits in the computational basis to obtain a bit string $\hat{z} \in \{0,1\}^m$ and declare $\hat{a} = \sin^2\left(\pi\left[\sum_{l=1}^m z_l 2^{l-1}\right] /2^m\right)$ to be an estimator of $a$. 

Typically, we define $M \equiv 2^m$ to be the number of iterations of the QAE algorithm. This is because the algorithm can be decomposed into a gate sequence of $M - 1$ controlled-$\mathcal{Q}$ gates. Crucially, the estimator $\hat{a}$ approximates the true value of $a$ with an error that scales as $\mathcal{O}(1/M)$. In contrast, simply estimating $a$ via the frequencies of outcomes for $M$ repeated measurements yields an estimation error that scales as $\mathcal{O}(1/\sqrt{M})$. 


\begin{figure*}[t!]
    \includegraphics[width=\textwidth]{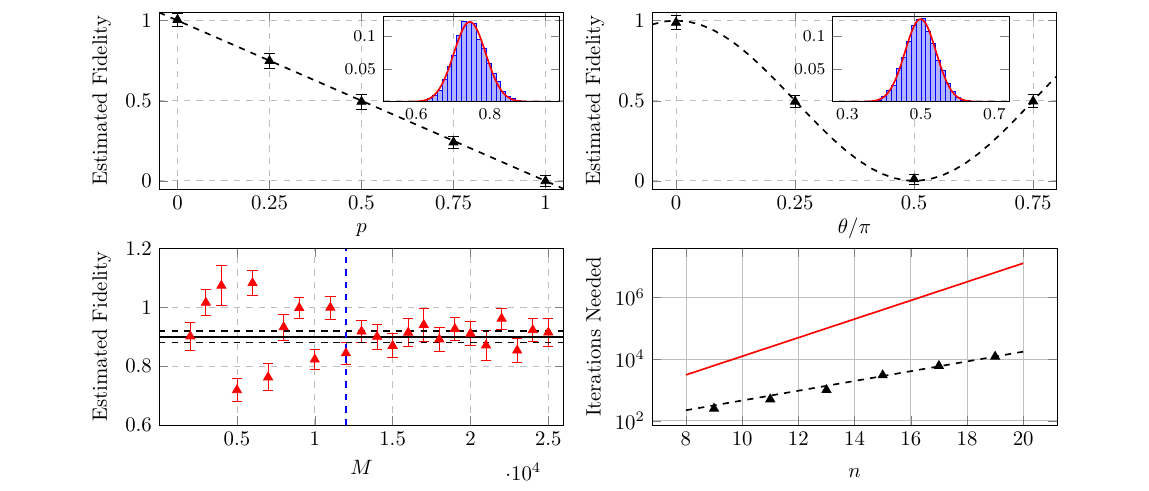}
    \caption{\textbf{Top}: Fidelity between the noisy and pure GHZ states estimated by simulations of our fidelity estimation protocol for depolarizing noise (left) and Pauli Z noise (right). Data points (triangles) correspond to an average of the fidelity estimators produced by $100$ runs of the protocol, while the error bars correspond to the standard deviation of the samples. The dashed lines indicate the true fidelity. In our simulations, we take $N = 1000$, $M = 500$, and $K = 10$. The insets in the left and right panels show the empirical distributions of fidelity estimators for $\theta = \pi/2$ and $p = 0.75$, respectively, each obtained from a sample of $5000$ data points. The red curves represent Gaussian fits. \textbf{Lower Left}: The mean and standard deviation of the estimated fidelity for various $M$. The vertical dashed line corresponds to the value of $M$ for which the fidelity estimation achieves the desired precision, as indicated by the horizontal dashed line (see main text for the procedure). \textbf{Lower Right}: Iterations $M$ needed to estimate fidelity with an accuracy of $0.02$ v.s. system size $n$. In the noisy state preparation of the GHZ state, simultaneous Pauli $Z$ errors can occur with probability $p = 0.1$. The illustrated data is for $N = 1000$ and $K = 10$. The dashed line is the best fit to $\alpha d^\beta$, with $\beta = 0.52$, indicating a $\mathcal{O}(\sqrt{d})$ scaling. The red line corresponds to $\alpha d$, emphasizing the advantage offered by a quadratic speedup.}
    \label{fig:simulation}
\end{figure*}

\vspace{8pt}
\emph{Fidelity Estimation Protocol.}---Suppose that two parties, Alice and Bob, are able to produce multiple copies of their states $|\psi\rangle$ and $\rho$, respectively. Our protocol executes the following steps.

\vspace{5pt}
\textbf{P1:} Bob performs classical shadow tomography with $N$ repetitions using random Clifford unitaries. He then sends the sampled data over to Alice, including the unitary $C_i$ and the measurement outcome $|\hat{b}_i\rangle$, with $1 \le i \le N$.

\textbf{P2:} For each $1\leq i \leq N$, Alice applies the QAE algorithm consisting of $M$ iterations independently $K$ times, to obtain the $K$ estimators $\hat{a}_{i, 1}, \dots, \hat{a}_{i,K}$ of the probability amplitude $a_i \equiv |\langle \hat{b}_i|\hat{C}_i|\psi\rangle|^2$. Next, for each $1\leq i \leq N$, Alice computes
\begin{equation}\label{eq:Fi-comp}
    \hat{F}_i \equiv \frac{(2^n + 1)}{N}\med(\{\hat{a}_{i,j}\}_{j=1}^K) - 1.
\end{equation}
Using this data, Alice computes a median-of-means estimator by splitting the $N$ values of $\hat{F}_i$ into $P$ equally sized parts, computing each part's mean, and then computing the median of the resulting collection of means: 
\begin{equation}\label{eq:Fmed-comp}
    \hat{F}_{\med} \equiv \med\left(\frac{1}{N/P} \sum_{l=(k-1)N/P + 1}^{kN/P}\hat{F}_i \right)_{k=1}^P.
\end{equation}
Finally, Alice declares $\hat{F}_{\med}$ to be an estimator of the true fidelity $F \equiv \langle \psi|\rho|\psi\rangle$.

Here, to implement QAE, Alice chooses the unitary $\mathcal{A}$ in~\eqref{eq:superposition} so that $\mathcal{A}|0\rangle^{\otimes n} = \hat{C}_i|\psi\rangle$ and takes $\chi:\{0,1\}^n \to \{0,1\}$ to be defined so that $\chi(z) = 1$ if and only if $z = \hat{b}_i$. The total number of operations Alice performs is $NMK$.

\vspace{5pt}
We briefly give an overview of why our protocol works. A straightforward computation [see Lemma~3 in the Supplemental Materials (SM)] using~\eqref{eq:Clifford-inverse} and~\eqref{eq:avg-of-inverse} shows that 
\begin{equation}\label{eq:F-expression}
    F = (2^n + 1) \mathbb{E}_{\hat{C}, \hat{b}}[|\langle \hat{b}|\hat{C}|\psi\rangle|^2] - 1.
\end{equation}
The idea is then to learn the value of $|\langle \hat{b}|\hat{C}|\psi\rangle|^2 $ approximately for many instances of $\hat{b}$ and $\hat{C}$ using QAE and then to statistically process the resulting data to obtain an accurate estimate of the right-hand side of~\eqref{eq:F-expression}. To this end, we consider the median of the amplitude estimates $\hat{a}_{i, 1}, \dots, \hat{a}_{i,K}$ in~\eqref{eq:Fi-comp} in order to ensure that QAE gives an accurate estimate with probability arbitrarily close to one; we employ median-of-means estimation in~\eqref{eq:Fmed-comp} to combat the spread in the distribution of $|\langle \hat{b}|\hat{C}|\psi\rangle|^2$. Note that the ordinary sample mean would have worked just as well for this purpose, though we can get a tighter formal guarantee using median-of-means. This leads to our main result.

\begin{theorem}\label{thm:main-thm} For any $\varepsilon, \delta \in (0,1)$, we have 
\begin{equation}
    \Pr(|\hat{F}_{\med} - F| > \varepsilon) \leq \delta 
\end{equation}
provided that the total number $N_A$ of QAE iterations that Alice employs and the number $N_B$ of state copies that Bob uses for CST satisfy
\begin{equation}
    N_A = \mathcal{O}\left(\frac{\ln(1/\varepsilon^2)\ln(1/\delta)\sqrt{d}}{\varepsilon^4}\right),
    N_B = \mathcal{O}\left(\frac{\ln(1/\delta)}{\varepsilon^2}\right).
\end{equation}   
\end{theorem}
This demonstrates an overall $\mathcal{O}(\sqrt{d})$ computational complexity ($N_A$) and a finite measurement and communication cost ($N_B$). Proof of the theorem is presented in the SM.

\vspace{8pt}
\emph{Simulations.}--- 
To build confidence in the performance of our protocol, we run a series of numerical experiments. We consider Alice's state $|\psi\rangle$ to be the $n$-qubit GHZ state, which we write as $|\GHZ\rangle$, and Bob's state $\rho$ to be a noisy preparation of $|\GHZ\rangle$. It is worth emphasizing that our protocol promises to work for any generic quantum state. We pick the GHZ state since it stays within the stabilizer formalism when rotated by Clifford unitaries, admitting efficient numerical simulations.

In more details, we first generate a random Clifford circuit $C$ using the algorithm described in~\cite{Bravyi2021} via the software package \textit{Stim}~\cite{Gidney2021}. Subsequently, we evolve $|\GHZ\rangle$ via $C$ and draw a bit string $b\in \{0,1\}^n$ according to the distribution of computational basis measurement outcomes on $C|\GHZ\rangle$. To simulate the QAE step, we wish to know the probability distribution of Alice's estimates for the amplitude $|\langle b|C|\GHZ\rangle|^2$. Fortunately, this distribution has a simple closed-form expression in terms of the number of iterations $M$ in the QAE algorithm and the true amplitude $|\langle b|C|\GHZ\rangle|^2$ (see the SM). Hence, we can simulate Alice's estimate of $|\langle b|C|\GHZ\rangle|^2$ by computing the \emph{exact} value of $|\langle b|C|\GHZ\rangle|^2$, constructing Alice's probability distribution based on of this value, and then sampling from it.

As would be expected, the distribution of Alice's estimates is tightly peaked for large values of $M$. Since the values of $M$ in some of the simulations we consider can be somewhat high, we truncate Alice's distribution to outcomes within a certain range of its peaks so that the sum of the probabilities of the remaining outcomes is greater than $0.999$ and renormalize. We sample from the resulting distribution in the interest of reducing computation time. Using this method, we simulate Alice's selection of $K$ estimators $\hat{a}_1, \dots, \hat{a}_K$ for the true amplitude $|\langle b|C|\GHZ\rangle|^2$. Using the median of these values, we compute the end value $(2^n + 1)\med(a_1, \dots, a_K) - 1$. We iterate through the aforementioned procedure $N$ times, randomly picking a Clifford unitary for each run, and then compute the average of the $N$ end values. We record this average as one of Alice's fidelity estimators. Note that in our numerical method, we opt to directly use the sample mean instead of partitioning the $N$ end values into equally sized partitions and computing their median-of-means (as described formally in Theorem~\ref{thm:main-thm}) because we observed no meaningful difference between the two numerically. 

We employed two strategies to prepare noisy versions of $|\GHZ\rangle$ on Bob's side. The first model consists of a 9-qubit GHZ state subjecting to a uniform Pauli Z error with probability $p$, i.e.,
\begin{equation}\label{eq:Z-error}
    \rho = (1-p)|\GHZ\rangle\langle \GHZ| + p Z^{\otimes 9}|\GHZ\rangle\langle \GHZ|Z^{\otimes 9}.
\end{equation}
For the second model, we studied a 8-qubit GHZ state that undergoes depolarizing (twirling channel). In this case, we have
\begin{equation}
    \rho = \int dC \hat{C}^\dagger e^{i\theta Z_1} \hat{C}|\GHZ\rangle\langle\GHZ|\hat{C}^\dagger e^{-i\theta Z_1}\hat{C},
\end{equation}
where the integration is with respect to the uniformly weighted measure on the Clifford group. 

For both noise models, Fig.~\ref{fig:simulation} (Top) demonstrated excellent agreements between the estimated and the true fidelities for different values of the noise parameters. Furthermore, the resources represented by $N$, $M$, and $K$ in our simulations are much lower than the thresholds at which Theorem~\ref{thm:main-thm} guarantees an accuracy $\varepsilon$ equal to the standard deviations of the empirical distributions. 

Finally, we plot the scaling of the number of iterations $M$ that Bob requires to estimate the fidelity with a given accuracy (defined below) for fixed $N$ and $K$ for the Pauli Z noise model~\eqref{eq:Z-error}. To accomplish this task, we compute the means of the empirical fidelity estimator distributions for varying values of $M$. We then choose the first value of $M$ such that for the majority of higher values of $M$, their error bars lie within the strip of radius $0.02$ about the true fidelity (see Fig.~\ref{fig:simulation}, Lower Left, and additional simulations in the SM)~\footnote{More sophisticated treatment can be employed to extract the value of $M$ for a given accuracy threshold. The informal criteria used here is sufficient to demonstrated a clear scaling of $M$ $v.s.$ system size.}. We plot these selected values of $M$ as a function of qubit count $n$ in the lower right panel of Fig.~\ref{fig:simulation}. The data clearly reveals that the values of $M$ scale approximately as $\sqrt{d}$. The desired $\mathcal{O}(\sqrt{d})$ scaling occurs despite again the choice of $N$ and $K$ smaller than that prescribed by Theorem~\ref{thm:main-thm} for $\varepsilon = 0.02$.

\vspace{8pt}
\emph{Discussion.}--- In this study, we introduced a protocol for estimating the fidelity between generic quantum states. The direct fidelity estimation protocol~\cite{Flammia2011, daSilva2011}, which is the best-known strategy for estimating the fidelity of two \textit{arbitrary} states, requires $\mathcal{O}(d)$ state copies. Our scheme requires only $\mathcal{O}(\sqrt{d})$ resources, offering a \textit{quadratic speedup}. Furthermore, both the number of quantum measurements and the amount of information shared between involved parties are reduced to \emph{finite}, significantly outperforming known methods.

We also remark that it is difficult to imagine a straightforward modification of CST that would result in our speedup. For instance, suppose that Bob samples $N$ classical snapshots $\hat{\rho}_1, \dots, \hat{\rho}_N$ of his state $\rho$. Then $F \approx \frac{1}{N} \sum_{i=1}^N \Tr(\hat{\rho}_i |\psi\rangle\langle \psi|)$. Since the snapshots $\hat{\rho}_i$ are observables with classical descriptions, Alice may attempt to apply CST via randomized measurements on her state $|\psi\rangle\langle \psi|$ to estimate the values $\Tr(\hat{\rho}_i |\psi\rangle\langle \psi|)$, i.e., getting classical snapshots $\hat{\psi}_j$ for $|\psi\rangle\langle \psi|$. However, per definition of the operator $\hat{\rho}_i$ in Eq.~(\ref{eq:Clifford-inverse}), the number of classical snapshots Alice needs to have a guaranteed performance (i.e, bound the estimation error to arbitrary precision) is $\mathcal{O}(d^2)$.

If one considers the quadratic speedup of DFE over full state tomography as ``statistical'' (importance sampling), the further quadratic speedup offered by our protocol can be viewed to have a ``quantum'' origin, as it essentially relies on quantum amplitude amplification. This speedup is also reminiscent of the optimal quadratic speedup in quantum sensing and unstructured quantum search~\cite{Bennett1997, Boyer1998}. It is therefore interesting to investigate whether the performance of our protocol is optimal as well. Since our fidelity estimation scheme already considers the introduction of a quantum algorithm, another possible direction is to investigate the advantage offered by using \textit{quantum communication} between Alice and Bob in our scenario.

\vspace{8pt}
Acknowledgement.---This work was supported in part by the U.S. Department of Energy, Office of Science, Office of Advanced Scientific Computing Research, through the Quantum Internet to Accelerate Scientific Discovery Program, and in part by the LDRD program at Los Alamos. C.V. acknowledges support from the Center for Nonlinear Studies. C.V. would also like to thank Faisal Alam and Shivan Mittal for helpful discussions.

\bibliography{references}

\clearpage
\appendix

\setcounter{page}{1}
\renewcommand\thefigure{\thesection\arabic{figure}}
\setcounter{figure}{0} 

\onecolumngrid

\begin{center}
\large{ Supplemental Material for \\ ``Direct Fidelity Estimation for Generic Quantum States''
}
\end{center}

\begin{center}
Christopher Vairogs$^{1,2}$ and Bin Yan$^{1}$\\
\small{$^{1}$\textit{Theoretical Division, Los Alamos National Laboratory, Los Alamos, New Mexico 87545, USA}}\\
\small{$^{2}$\textit{Department of Physics, University of Illinois at Urbana-Champaign, Urbana, Illinois 61801, USA}}\\
\small{(Dated: \today)}
\end{center}

Unless otherwise stated, all states are assumed to belong to an $n$-qubit Hilbert space $(\mathbb{C}^2)^n$. Write $d\equiv 2^n$. Let $\mathcal{D}((\mathbb{C}^2)^n)$ and $\mathcal{U}((\mathbb{C}^2)^n)$ denote the spaces of density and unitary operators over $(\mathbb{C}^2)^n$, respectively. For any pure state $|\varphi\rangle \in (\mathbb{C}^2)^n$, we use the notational convention that $\varphi \equiv |\varphi\rangle\langle \varphi|$. The fidelity between two states $\rho, \sigma \in \mathcal{D}((\mathbb{C}^2)^n)$ is defined as 
\begin{equation}
    F(\rho, \sigma) \equiv \Tr\left[\sqrt{\sqrt{\rho}\sigma\sqrt{\rho}}\right]^2,
\end{equation}
so that for any $|\varphi\rangle \in (\mathbb{C}^2)^n$ we have $F(\rho, \varphi) = \langle \varphi|\rho|\varphi\rangle$.

\begin{section}{Review of Classical Shadow Tomography}

In the typical set-up of classical shadow tomography, one samples a unitary $\hat{U}$, where the hat symbol is used to indicate a random variable, from some ensemble $\mathcal{U}$ of unitaries and evolves $\rho$ via $\hat{U}$. The resulting state is subsequently measured in the computational basis $\{|b\rangle: b\in \{0,1\}^n\}$, obtaining outcome $|\hat{b}\rangle$. The above procedure is repeated many times. The obtained data (including the sampled unitaries and the measured bit strings) are recorded and can be used to predict many properties of the original quantum state $\rho$. The precision of this approach depends on the choice of random unitary ensemble together with the type of observables to be evaluated. Here, we are interested in predicting the fidelity of a pure state with respect to $\rho$ (i.e., the observable is the projector $\psi$). In this case, one typically samples the random unitary from the Clifford group. We present the technical aspects for this case in the following.

Define the quantum channel $\mathcal{M}: \mathcal{D}((\mathbb{C}^2)^n) \to \mathcal{D}((\mathbb{C}^2)^n)$ as follows. Let $\sigma\in \mathcal{D}((\mathbb{C}^2)^n)$ be an arbitrary state. Consider an operator $\hat{C}$ sampled randomly from the uniformly weighted Clifford group $\mathcal{C}_n$. Suppose the rotated state $\hat{C}\sigma \hat{C}^\dagger$ is measured in the computational basis $\{|b\rangle\}_{b\in \{0,1\}^n}$, yielding a post-measurement state $|\hat{b}\rangle$. Treating $\hat{C}$ and $|\hat{b}\rangle$ as random variables, we define $\mathcal{M}(\sigma)$ to be the expectation value of $\hat{C}^\dagger |\hat{b}\rangle\langle \hat{b}|\hat{C}$ over the distributions of both $\hat{C}$ and $|\hat{b}\rangle$. That is, 
    \begin{equation}
        \mathcal{M}(\sigma) = \sum_{b\in \{0,1\}^n} \mathbb{E}_{\hat{C} \sim \mathcal{C}_n} [\hat{C}^\dagger |b\rangle\langle b|\hat{C} \times \langle b| \hat{C}\sigma \hat{C}^\dagger|b\rangle]. 
    \end{equation}

The following proposition is a standard result of classical shadow tomography~\cite{Huang2020}.
\begin{proposition}[Classical Shadow Tomography]\label{prop:CST}
    Let $\rho\in \mathcal{D}((\mathbb{C}^2)^n)$ be an arbitrary state. 
    Take $\hat{C}$ again to be sampled randomly from the uniformly weighted Clifford group $\mathcal{C}_n$. Suppose that the rotated state $\hat{C}\rho \hat{C}^\dagger$ obtained from $\rho$ is measured in the computational basis $\{|b\rangle\}_{b\in \{0,1\}^n}$, yielding post-measurement state $|\hat{b}\rangle$. The following hold.
    \begin{enumerate}
        \item The map $\mathcal{M}$ is invertible and $\mathcal{M}^{-1}$ gives the following equality of random variables: 
        \begin{equation}
            \mathcal{M}^{-1}(\hat{C}^\dagger |\hat{b}\rangle\langle \hat{b}| \hat{C}) = (d+1)\hat{C}^\dagger|\hat{b}\rangle\langle\hat{b}|\hat{C} - I_n. 
        \end{equation}
        \item Let $O$ be a Hermitian operator over $(\mathbb{C}^2)^n$. If we define the random variable $\hat{o} \equiv \Tr\left[\mathcal{M}^{-1}(\hat{C}^\dagger|\hat{b}\rangle\langle \hat{b}|\hat{C}) O\right]$, then
        \begin{equation}
            \mathbb{E}[\hat{o}] = \Tr[O\rho]
        \end{equation}
        and
        \begin{equation}
            \var[\hat{o}] \leq 3 \Tr[O^2].
        \end{equation}
    \end{enumerate}        
\end{proposition}
Note that while $\mathcal{M}$ is a quantum channel and $\mathcal{M}^{-1}$ is well-defined, the map $\mathcal{M}^{-1}$ is not itself a channel.

\begin{lemma}\label{lemma:overlap-statistics}
    Let $\rho\in \mathcal{D}((\mathbb{C}^2)^n)$ be a state that is not-necessarily pure and let $|\psi\rangle \in (\mathbb{C}^2)^n$ be a pure state. Let $\hat{C}$ be an operator sampled randomly from the uniformly weighted Clifford group $\mathcal{C}_n$. Consider that the rotated state $\hat{C}\rho \hat{C}^\dagger$ is measured in the computational basis $\{|b\rangle\}_{b\in \{0,1\}^n}$, yielding post-measurement state $|\hat{b}\rangle$. Then 
    \begin{equation}\label{eq:overlap-expectation}
        \mathbb{E}\left[|\langle \hat{b}|\hat{C}|\psi\rangle|^2\right] = \frac{1+ F(\rho, \psi)}{d+1}
    \end{equation}
    and
    \begin{equation}\label{eq:overlap-variance}               \var\left[|\langle\hat{b}|\hat{C}|\psi\rangle|^2\right] \leq \frac{3}{(d+1)^2}.
    \end{equation}
\end{lemma}

\begin{proof} By Proposition~\ref{prop:CST}, we have that
\begin{equation}
    F(\rho, \psi) = \Tr[\rho\psi] = \mathbb{E}\left[\Tr\left[\mathcal{M}^{-1}(\hat{C}^\dagger|\hat{b}\rangle\langle\hat{b}|\hat{C}) \psi \right]\right] = (d+1)\mathbb{E}\left[\Tr\left[\hat{C}^\dagger|\hat{b}\rangle\langle\hat{b}|\hat{C} \psi\right]\right] - 1 = (d+1)\mathbb{E}\left[|\langle \hat{b}|\hat{C}|\psi\rangle|^2\right] - 1.
\end{equation}
Equation~\eqref{eq:overlap-expectation} follows. By Proposition~\ref{prop:CST} and a property of the variance, we also have that 
\begin{align}    
\var\left[|\langle \hat{b}|\hat{C}|\psi\rangle|^2\right] &= \var\left[\frac{1}{d+1}
    \Tr\left[\mathcal{M}^{-1}(\hat{C}^\dagger|\hat{b}\rangle\langle\hat{b}|\hat{C})\psi\right] + \frac{1}{d+1}\right]\\  &= \left(\frac{1}{d+1}\right)^2 \var\left[\Tr\left[\mathcal{M}^{-1}(\hat{C}^\dagger|\hat{b}\rangle\langle\hat{b}|\hat{C})\psi\right]\right] \\ &\leq 
    \left(\frac{1}{d+1}\right)^2 \left(3 \Tr\left[\psi^2\right]\right) \\
    &= \frac{3}{(d+1)^2}.
\end{align} \end{proof}

\end{section}

\begin{section}{Some Concentration Inequalities}

\begin{proposition}[Hoeffding's inequality]
    Let $\hat{X}_1, \dots, \hat{X}_N$ be independent random variables with the property that $a_i \leq \hat{X}_i \leq b_i$ almost surely for $1\leq i \leq N$. Then for any $\varepsilon > 0$
    \begin{equation}
        \Pr\left( \sum_{i=1}^N \hat{X}_i - \mathbb{E}\left[\sum_{i=1}^N \hat{X}_i \right] \geq \varepsilon \right) \leq \exp\left(\frac{-2\varepsilon^2}{\sum_{i=1}^N (b_i - a_i)^2} \right)
    \end{equation}
    and
    \begin{equation}
        \Pr\left( \left|\sum_{i=1}^N \hat{X}_i - \mathbb{E}\left[\sum_{i=1}^N \hat{X}_i \right] \right| \geq \varepsilon \right) \leq 2\exp\left(\frac{-2\varepsilon^2}{\sum_{i=1}^N (b_i - a_i)^2} \right).
    \end{equation}
\end{proposition}
Hoeffding's inequality is useful for bounding the probability that a sample mean deviates from the true mean of the underlying distribution by a small amount. However, sometimes we wish to bound deviations of a data sample's measure of central tendency from a value that is not guaranteed to be the mean of the underlying distribution. In this scenario, the following lemma is desirable.

\begin{lemma}\label{lemma:median-concentration}
    Let $a\in \mathbb{R}$ and let $\hat{a}_1, \dots, \hat{a}_N$ be independent random variables such that $\Pr(|\hat{a}_i - a| > \varepsilon) \leq \delta<1/2$ for some $\varepsilon > 0$ and $\frac{1}{2} > \delta > 0$. Then 
    \begin{equation}
        \Pr(|\med(\{\hat{a}_i\}_{i=1}^N) - a| > \varepsilon) \leq \exp\left(-2\left(\frac{1}{2} - \delta \right)^2 N\right).
    \end{equation}
\end{lemma}

\begin{proof} The condition that $|\med(\{\hat{a}_i\}_{i=1}^N) - a| > \varepsilon$ is equivalent to the condition that $|\hat{a}_i - a| > \varepsilon$ for at least $N/2$ of the $\hat{a}_i$. Since $\Pr(|\hat{a}_i - a| > \varepsilon) \leq \delta$ for $1\leq i \leq N$, it follows that 
\begin{align}
    \Pr(|\med(\{\hat{a}_i\}_{i=1}^N) - a| > \varepsilon) 
    &\leq 
    \Pr\left(\hat{B}(N, \delta) \geq \frac{N}{2}\right) \\
    &=
    \Pr\left(\hat{B}(N, \delta) - N\delta \geq N\left(\frac{1}{2} - \delta \right) \right) \\
    &=
    \Pr\left(\hat{B}(N, \delta) - \mathbb{E}\left[\hat{B}(N, \delta)\right] \geq N \left(\frac{1}{2} - \delta\right)\right) \\
    &\leq\label{eq:Hoeffding-application-med} \exp\left(\frac{-2N^2\left(\frac{1}{2} - \delta\right)^2}{N}\right) \\
    &= \exp\left(-2\left(\frac{1}{2} - \delta\right)^2 N\right),
\end{align}
where $\hat{B}(N, \delta)$ is a binomial random variable with $N$ trials and success probability $\varepsilon$. Inequality~\eqref{eq:Hoeffding-application-med} is obtained by viewing $\hat{B}(N, \delta)$ as the sum of $N$ independent (Bernoulli) random variables that assume value 1 with probability $\delta$ and 0 with probability $1 - \delta$ and applying Hoeffding's inequality.
\end{proof}

\end{section}

\begin{section}{Review of Quantum Amplitude Estimation}

Fix a state $|\Psi\rangle \in (\mathbb{C}^2)^n$. Let $m$ be a positive integer and define $M\equiv 2^m$. Choose an orthonormal basis $\{|z\rangle_n\}_{z = 0}^{2^n-1}$ for $(\mathbb{C}^2)^n$. Choose another orthonormal basis $\{|z\rangle_m\}_{z = 0}^{M-1}$ for $(\mathbb{C}^2)^m$. Consider a boolean function $\chi: \{0, \dots, 2^n-1\} \to \{0,1\}$. Then $|\Psi\rangle$ can be uniquely written as $|\Psi\rangle = |\Psi_0\rangle + |\Psi_1\rangle$, where $|\Psi_0\rangle$ and $|\Psi_1\rangle$ are linear combinations of states $|z\rangle_n$ with $z\in \{0, \dots, 2^n-1\}$ for which $\chi(z) = 0$ and $\chi(z) = 1$, respectively. The vectors $|\Psi_0\rangle$, $|\Psi_1\rangle$ are not necessarily normalized. The goal of the quantum amplitude estimation algorithm~\cite{Brassard2002} of Brassard, Hoyer, Mosca and Tapp is to estimate the probability amplitude $a \equiv \langle \Psi_1|\Psi_1\rangle$. The algorithm assumes access to a unitary $\mathcal{A}$ that prepares the state $|\Psi\rangle$ from the state $|0\rangle_n$. That is, $\mathcal{A}|0\rangle_n = |\Psi\rangle$.

\begin{protocol}[Quantum Amplitude Estimation Protocol]\label{qae-protocol}
The protocol goes as follows.
\begin{enumerate}
    \item Initialize a state register of $n$ qubits and another register of $m$ ancilla qubits, all in the $|0\rangle$ state. 
    
    \item On the state register, apply the unitary $\mathcal{A}$. On the ancilla qubits, apply a quantum Fourier transform $\QFT_M$.
    
    \item Define the $n$-qubit unitaries $\mathcal{S}_0 \equiv I_n - 2|0\rangle\langle 0|_n$ and 
    \begin{equation}
        \mathcal{S}_\chi \equiv I_n - 2\sum_{\stackrel{z\in \{0, \dots, 2^n - 1\}:}{ \chi(z) = 1}}|z\rangle\langle z|_n.
    \end{equation}
    Define $\mathcal{Q} \equiv -\mathcal{A}\mathcal{S}_{0}\mathcal{A}^\dagger \mathcal{S}_{\chi}$. Apply the gate $\Lambda_M(\mathcal{Q})$ simultaneously to the ancilla and state registers, where 
    \begin{equation}
        \Lambda_M(\mathcal{Q}) \equiv \sum_{j=0}^{M-1} |j\rangle\langle j|_m \otimes \mathcal{Q}^j. 
    \end{equation}
    Note that $\Lambda_M(\mathcal{Q})$ may be decomposed into a sequence of two-qubit controlled-$\mathcal{Q}$ gates as mentioned in the main text. 
    \item Apply the inverse quantum Fourier transform $\QFT_M^\dagger$ on the ancilla register. 
    \item Measure the ancilla qubits in the basis $\{|z\rangle_m\}_{z=0}^{M-1}$ to obtain an integer $\hat{y}\in \{0, M-1\}$.
    \item Compute $\hat{a} = \sin^2(\pi \hat{y}/M)$ and declare $\hat{a}$ to be an estimator of $a$. 
\end{enumerate}
\end{protocol}
The number of iterations of the algorithm described above is \textit{defined} to be $M = 2^m$. For an illustration of the algorithm, see Fig. 1 of~\cite{Brassard2002}. The following proposition comes from Theorem 12 of the same paper. 

\begin{proposition}\label{prop:qae-guarantee}
    For the amplitude estimation protocol (Protocol~\ref{qae-protocol}), we have
    \begin{equation}
        \Pr\left(|\hat{a} - a| > \frac{2\pi \sqrt{a(1-a)}}{M} + \frac{\pi^2}{M^2}\right) \leq 1- \frac{8}{\pi^2} \approx 0.1894.
    \end{equation}
\end{proposition}

\end{section}

\begin{section}{Fidelity Estimation Protocol}

Suppose Alice is able to prepare multiple copies of a state $|\psi\rangle \in (\mathbb{C}^2)^n$ from the state $|0\rangle^{\otimes n}$ via a unitary $\mathcal{A}$, as in Protocol~\ref{qae-protocol}, and Bob has access to multiple copies of a state $\rho \in \mathcal{D}(\mathbb{C}^2)^n$ that is not-necessarily pure. We assume that neither has any knowledge of the other's state, nor that either knows a classical description of their own state. We propose the following protocol for Alice and Bob to estimate the fidelity of their states $\psi$ and $\rho$.

\begin{protocol}[Fidelity Estimation Protocol]\label{fid-est-protocol}
The protocol goes as follows.
    \begin{enumerate}
        \item Bob draws $N$ independent and identically distributed samples $\hat{C}_1, \dots, \hat{C}_N$ from the uniformly weighted Clifford group $\mathcal{C}_n$. He then rotates $N$ copies of her state $\rho$ according to her sampled Clifford operators, obtaining the states $\hat{C}_1 \rho \hat{C}_1^\dagger, \dots, \hat{C}_N \rho \hat{C}_N^\dagger$. He measures each rotated state $\hat{C}_i \rho \hat{C}_i^\dagger$ in the computational basis $\{|b\rangle\}_{b\in \{0,1\}^N}$, obtaining the post-measurement state $|\hat{b}_i\rangle$ on the state $\hat{C}_i \rho \hat{C}_i^\dagger$. Thus, 
        \begin{equation}
            \Pr(\hat{C}_i = C~{\rm and}~\hat{b}_i = b) = \frac{\langle b|C\rho C^\dagger |b\rangle}{|\mathcal{C}_n|}.
        \end{equation}
        Subsequently, he sends a classical description of her sampled data $\hat{C}_i$, $|\hat{b}_i\rangle$, $1\leq i \leq N$, over to Bob.
        \item For each $i \in \{1, \dots, N\}$, Alice estimates the probability amplitude $\hat{A}_{i}\equiv |\langle \hat{b}_i|\hat{C}_i|\psi\rangle|^2$ as follows. Fix a one-to-one correspondence between the bit strings $\{0,1\}^n$ and integers $\{0, \dots, 2^{n}-1\}$. Alice uses $K$ independent applications of the quantum amplitude estimation algorithm (Protocol~\ref{qae-protocol}) with $|\Psi\rangle = \hat{C}_i|\psi\rangle$ and $\chi:\{0, \dots, 2^n - 1\}\to \{0,1\}$ defined so that $\chi$ assumes value 1 only for the unique integer corresponding to $\hat{b}_i$. Each application consists of $M$ iterations. As a result of the $K$ independent applications of the amplitude estimation algorithm, Alice obtains a set of $K$ estimators $\hat{a}_{i, 1}, \dots, \hat{a}_{i, K}$ of $\hat{A}_i$.  (Thus, the total number of operations he performs in this step is $N\times K\times M$.)
        \item Alice computes
        \begin{equation}
            \hat{F} \equiv \frac{(d+1)}{N} \sum_{i=1}^N \med(\{\hat{a}_{i,j}\}_{j=1}^K) - 1
        \end{equation}
        and declares $\hat{F}$ to be an estimator of the true fidelity $F(\rho, \psi)$.
    \end{enumerate}
\end{protocol}

\begin{proposition}[Protocol~\ref{fid-est-protocol} Performance Guarantee] \label{prop:performance-guarantee} 
    Let $\varepsilon, \delta \in (0,1)$ be such that $\left\lceil\frac{24}{\varepsilon^2\delta} \right\rceil \leq \left(\frac{13}{6}\right)^4\frac{\delta}{12\varepsilon^4}$. Suppose we choose $N$, $M$, and $K$ in Protocol~\ref{fid-est-protocol} to satisfy
    \begin{equation}\label{eq:state-copy-bounds}
        M \geq \frac{2\pi \sqrt{3(d+1)}}{(6\varepsilon/13)^2}, \ \ \ \,  \frac{24}{\varepsilon^2\delta} \leq N \leq \left(\frac{13}{6}\right)^4\frac{\delta}{12\varepsilon^4}, \ \ \ \,  K \geq \frac{1}{2\left(\frac{8}{\pi^2} - \frac{1}{2}\right)^2}\ln\left(\frac{4N}{\delta} \right).
    \end{equation}
    Then
    \begin{equation}
        \Pr(|\hat{F} - F(\rho, \psi)| > \varepsilon) \leq \delta.
    \end{equation}
    The minimum total number of basic operations between Alice and Bob that satisfies the bounds in~\eqref{eq:state-copy-bounds} is
    \begin{equation}
        N_{\rm tot} = \mathcal{O}\left(\frac{1}{\varepsilon^2\delta} + \frac{\sqrt{d}}{\varepsilon^4\delta} \ln \left(\frac{1}{\varepsilon^2\delta^2}\right)\right).
    \end{equation}
\end{proposition}

While the restriction $\lceil 24/(\varepsilon^2\delta) \rceil \leq (13/6)^4 \delta/(12\varepsilon^2)$ in the proposition might not appear to be very transparent at first, it may easily be satisfied by many choices of $\varepsilon$ and $\delta$ desired in practical settings. For example, for $\varepsilon<1$, setting $\delta = 24\varepsilon$ will satisfy this bound. Furthermore, the thresholds for $N, K, M$ in the proposition were chosen for the sake of convenience, not necessarily because they are tight. With a more clever method of data processing and more careful selection of bounds, it seems likely that one could obtain a rigorous performance guarantee that does not require upper bounding $N$, and hence, eliminates the restriction that $\lceil 24/(\varepsilon^2\delta) \rceil \leq \delta/(12\varepsilon^4)$.

We now prove Proposition~\ref{prop:performance-guarantee}.

\begin{proof} Define the random variable $\hat{Y} \equiv \frac{d+1}{N} \left(\sum_{i=1}^N \hat{A}_i\right) - 1$. Note that $\mathbb{E}[\hat{Y}] = F(\rho, \psi)$ by Lemma~\ref{lemma:overlap-statistics}. Then we may apply Chebyshev's inequality to get
\begin{equation}
    \Pr\left(|\hat{Y} - F(\rho, \psi)| > \frac{\varepsilon}{2}\right) \leq \frac{4 \var(\hat{Y})}{\varepsilon^2} = \frac{4 (d+1)^2 \var\left(\frac{1}{N} \sum_{i=1}^N \hat{A}_i\right)}{\varepsilon^2} 
    = \frac{4(d+1)^2 \var(|\langle \hat{b}_1|\hat{C}_1|\psi\rangle|^2)}{N\varepsilon^2} \leq \frac{12}{N\varepsilon^2} \leq \frac{\delta}{2},
\end{equation}
where to obtain the second inequality, we used~\eqref{eq:overlap-variance}.

We now bound the difference between $\hat{Y}$ and our fidelity estimator $\hat{F}$. We have
\begin{align}
    \Pr\left(|\hat{F} - \hat{Y}| >\label{eq:f-y-diff} \frac{\varepsilon}{2}\right) &=
    \Pr\left(\left| \sum_{i=1}^N (\med(\{\hat{a}_{i,j}\}_{j=1}^K) - \hat{A}_i) \right| > \frac{N \varepsilon}{2(d+1)} \right) \\ 
    &\leq \Pr\left( \sum_{i=1}^N \left|\med(\{\hat{a}_{i,j}\}_{j=1}^K) - \hat{A}_i\right| > \frac{N \varepsilon}{2(d+1)} \right) \\
    &\leq\label{ineq:union-bound} \sum_{i=1}^N \Pr\left(\left|\med(\{\hat{a}_{i,j}\}_{j=1}^K) - \hat{A}_i \right| > \frac{\varepsilon}{2(d+1)}\right).
\end{align}
See~\footnote{It may be tempting to apply Hoeffding's inequality here. However, we have no guarantee that 
\begin{equation}
    \mathbb{E}\left[\sum_{i=1}^N \med(\{\hat{a}_{i,j}\}_{j=1}^K) \right] = \sum_{i=1}^N \hat{A}_i,
\end{equation}
where the expectation value is over the distribution produced by the measurements in the amplitude estimation protocol. This leads us to use a combination of the triangle inequality and a union bound instead, giving a rather loose bound. 
} for comments on this step.

Recall that
\begin{equation}
    M \geq \frac{2\pi\sqrt{3(d+1)}}{\eta^2}, \quad \eta \equiv \frac{6\varepsilon}{13}.
\end{equation}
by assumption and define
\begin{equation}
    E \equiv \frac{1+ F(\rho, \psi)}{d+1}.
\end{equation}
Then 
\begin{align}
    \Pr\left(\underbrace{\left|\med(\{\hat{a}_{i,j}\}_{j=1}^K) - \hat{A}_i \right| > \frac{\varepsilon}{2(d+1)}}_{Z}\right) &= 
    \Pr\left(Z {\rm \ and \ } 
    |\hat{A}_i - E| \leq \frac{1}{\eta^2(d+1)}\right) +
    \Pr\left(Z {\rm \ and \ }
    |\hat{A}_i - E| > \frac{1}{\eta^2(d+1)} \right) \\
    &\leq \Pr\left(Z {\rm \ and \ } 
    |\hat{A}_i - E| \leq \frac{1}{\eta^2(d+1)}\right) +
    \Pr\left(|\hat{A}_i - E| > \frac{1}{\eta^2(d+1)}\right) \\
    &\leq\label{eq:Chebyshev-application} \Pr\left(Z {\rm \ and \ } 
    |\hat{A}_i - E| \leq \frac{1}{\eta^2(d+1)}\right) + 
    \frac{\frac{3}{(d+1)^2}}{\frac{1}{\eta^4(d+1)^2}} \\
    &= \Pr\left(Z {\rm \ and \ } 
    |\hat{A}_i - E| \leq \frac{1}{\eta^2(d+1)}\right) + 3\eta^4 \\
    &= \Pr\left(Z {\rm \ and \ } 
    |\hat{A}_i - E| \leq \frac{1}{\eta^2(d+1)}\right) + 3\left(\frac{6}{13}\right)^4 \varepsilon^4.
\end{align}
Chebyshev's inequality is used to obtain~\eqref{eq:Chebyshev-application}. Note that $|\hat{A}_i - E| \leq 1/[\eta^2(d+1)]$ implies that $\hat{A}_i \leq E + 1/[\eta^2(d+1)]$.

Hence, $|\hat{A}_i - E| \leq 1/[\eta^2(d+1)]$ implies that
\begin{align}
    \frac{2\pi \sqrt{\hat{A}_i (1 - \hat{A}_i)}}{M} + \frac{\pi^2}{M^2} &\leq \frac{\eta^2 \sqrt{\hat{A}_i(1-\hat{A}_i})}{\sqrt{3(d+1)}} + \frac{\eta^4}{12(d+1)}\\
    &\leq 
    \frac{\eta^2{\sqrt{\hat{A}_i}}}{\sqrt{3(d+1)}} + \frac{\eta^4}{12(d+1)} \\
    &\leq 
    \frac{\eta^2}{\sqrt{3(d+1)}}\sqrt{\frac{1+F(\rho, \psi)}{d+1} + \frac{1}{\eta^2(d+1)}} + \frac{\eta^4}{12(d+1)} \\
    &\leq 
    \frac{\eta^2}{\sqrt{3(d+1)}}
    \sqrt{\frac{3}{\eta^2(d+1)}} + 
    \frac{\eta^4}{12(d+1)} \\
    &\leq \frac{\eta}{d+1} + \frac{\eta^4}{12(d+1)} \\
    &\leq \frac{13\eta}{12(d+1)} \\
    &= \frac{\varepsilon}{2(d+1)}.
\end{align}
Hence, we get that
\begin{align}
    \Pr(Z) &=\label{eq:pr-z} \Pr\left(\left|\med(\{\hat{a}_{i,j}\}_{j=1}^K) - \hat{A}_i\right| > \frac{\varepsilon}{2(d+1)}~{\rm and}~|\hat{A}_i - E| \leq \frac{1}{\eta^2(d+1)} \right) + 3\left(\frac{6}{13}\right)^4\varepsilon^4 \\
    &\leq \Pr\left(\left|\med(\{\hat{a}_{i,j}\}_{j=1}^K) - \hat{A}_i\right| > \frac{2\pi\sqrt{\hat{A}_i(1-\hat{A}_i)}}{M} + \frac{\pi^2}{M^2} \right) + 3\left(\frac{6}{13}\right)^4\varepsilon^4 \\
    &\leq\label{eq:application-of-median-conc} \exp\left(-2\left(\frac{8}{\pi^2} - \frac{1}{2}\right)^2 K\right) + 3\left(\frac{6}{13}\right)^4\varepsilon^4,
\end{align}
where~\eqref{eq:application-of-median-conc} follows from Lemma~\ref{lemma:median-concentration} and Proposition~\ref{prop:qae-guarantee}. 

Substituting~\eqref{eq:application-of-median-conc} into~\eqref{ineq:union-bound}, we get that
\begin{equation}
    \Pr\left(|\hat{F} - \hat{Y}| > \frac{\varepsilon}{2}\right) \leq N \exp\left(-2\left(\frac{8}{\pi^2} - \frac{1}{2}\right)^2 K\right) + 
    3\left(\frac{6}{13}\right)^4 N\varepsilon^4 \leq \frac{\delta}{4} + \frac{\delta}{4} = \frac{\delta}{2}.
\end{equation}
It then follows that
\begin{align}
    \Pr(|\hat{F} - F(\rho, \psi)| > \varepsilon) &\leq \Pr(|\hat{F} - \hat{Y}| + |\hat{Y} - F(\rho, \psi)| > \varepsilon) \\
    &\leq \Pr(|\hat{F} - \hat{Y}|> \frac{\varepsilon}{2}) + \Pr(|\hat{Y} - F(\rho, \psi)| > \frac{\varepsilon}{2}) \\
    &\leq \frac{\delta}{2} + \frac{\delta}{2} \\ 
    &= \delta.
\end{align}

Finally, since Bob uses $N$ copies of his state, and since for each of Bob's $N$ measurement results, Alice uses $KM$ basic operations, the total number of basic operations between Alice and Bob in this protocol is $N + NKM$. If we choose the ceiling of the lower threshold for $N, K, M$ provided by~\eqref{eq:state-copy-bounds}, we get that
\begin{equation}
        N_{\rm tot} = \mathcal{O}\left(\frac{1}{\varepsilon^2\delta} + \frac{\sqrt{d}}{\varepsilon^4\delta} \ln \left(\frac{1}{\varepsilon^2\delta^2}\right)\right).
    \end{equation}

\end{proof}

We note that Alice and Bob may mitigate outlier corruption in their fidelity estimate by performing median-of-means estimation. That is, they repeat Protocol~\ref{fid-est-protocol} a number $P$ times, using $N$ of Alice's state copies, $M$ iterations of the quantum amplitude estimation algorithm, and $K$ independent applications of the amplitude estimation algorithm in each run. As a result, they get $P$ fidelity estimators $\hat{F}_1, \dots, \hat{F}_P$. They then use $\hat{F}_{\med} \equiv \med(\{\hat{F}_i\}_{i=1}^P)$ as their estimate for the true fidelity $F(\rho, \psi)$. Using this median-of-means strategy, we get a guarantee for their fidelity estimator with a faster theoretical convergence:

\begin{corollary}
    For any $\varepsilon \in (0, 1)$ and $\delta \in (0, 0.09)$, Alice and Bob may employ a median-of-means estimation strategy in which Bob uses
    \begin{equation}
        \mathcal{O}\left(\frac{\ln(1/\delta)}{\varepsilon^2}\right)
    \end{equation}
    state copies and Alices uses
    \begin{equation}
        \mathcal{O}\left(\frac{\ln(1/\epsilon^2) \ln(1/\delta)\sqrt{d}}{\varepsilon^4}\right)
    \end{equation}
    total iterations of the basic quantum amplitude estimation algorithm to get
    \begin{equation}
        \Pr(|\hat{F}_{\med} - F(\rho, \psi)| > \varepsilon) \leq \delta.
    \end{equation}
\end{corollary}

\begin{proof}
    Let $\varepsilon \in (0,1)$ and $\delta \in (0, 0.09)$ be arbitrary. Choose $\delta_0 = 1/3$. Since $\delta <0.09$, we have $\lceil 24/(\varepsilon^2 \delta_0)\rceil \leq (13/6)^4 \delta_0/(12 \varepsilon^4)$. It follows from Proposition~\ref{prop:performance-guarantee} that by choosing
    \begin{equation}
        M = \left\lceil \frac{2\pi \sqrt{3(d+1)}}{(6\varepsilon/13)^2}\right\rceil, \ \ \ \ N = \left \lceil \frac{24}{\varepsilon^2\delta_0} \right \rceil = \left \lceil \frac{72}{\varepsilon^2} \right \rceil, \ \ \ \ K = \left \lceil \frac{1}{2\left(\frac{8}{\pi^2} - \frac{1}{2}\right)^2} \ln \left(\frac{4N}{\delta_0} \right) \right \rceil = \left \lceil \frac{1}{2\left(\frac{8}{\pi^2} - \frac{1}{2}\right)^2} \ln \left(12\left\lceil \frac{72}{\varepsilon^2} \right\rceil \right) \right \rceil
    \end{equation}
    in Protocol~\ref{fid-est-protocol}, we get a fidelity estimator $\hat{F}$ for which
    \begin{equation}
        \Pr(|\hat{F} - F(\rho, \psi)| > \varepsilon) \leq \delta_0 = \frac{1}{3}.
    \end{equation}
    The $P$ fidelity estimators $\hat{F}_1, \dots, \hat{F}_P$ obtained by repeating Protocol~\ref{fid-est-protocol} a number $P$ times with this choice of $N$, $K$, and $M$ all satisfy this concentration inequality. Since $\delta_0 < 1/2$, Lemma~\ref{lemma:median-concentration} then reveals that
    \begin{equation}
        \Pr(|\hat{F}_{\med} - F(\rho, \psi)| > \varepsilon) \leq \exp\left(-2 \left(\frac{1}{2} - \delta_0 \right)^2 P\right) = \exp\left( - \frac{P}{18} \right).
    \end{equation}
    Consequently, by choosing $P = \lceil 18 \ln(1/\delta) \rceil$, we get that $\Pr(|\hat{F}_{\med} - F(\rho, \psi)| > \varepsilon) \leq \delta$. The total number of state copies that Bob uses across all applications of Protocol~\ref{fid-est-protocol} is
    \begin{equation}
        P \times N = \left \lceil 18 \ln(1/\delta) \right \rceil \left\lceil \frac{72}{\varepsilon^2} \right\rceil = \mathcal{O}\left(\frac{\ln(1/\delta)}{\varepsilon^2}\right).
    \end{equation}
    The total number of iterations that Alice uses is 
    \begin{equation}
        P \times N \times K \times M = \left \lceil 18 \ln\left(\frac{1}{\delta}\right) \right \rceil \left \lceil \frac{72}{\varepsilon^2} \right \rceil \left \lceil \frac{1}{2\left(\frac{8}{\pi^2} - \frac{1}{2}\right)^2} \ln \left(12\left\lceil \frac{72}{\varepsilon^2} \right\rceil \right) \right \rceil \left\lceil \frac{2\pi \sqrt{3(d+1)}}{(6\varepsilon/13)^2}\right\rceil = \mathcal{O}\left(\frac{\ln(1/\varepsilon^2) \ln(1/\delta) \sqrt{d}}{\varepsilon^4} \right).
    \end{equation} 
    
\end{proof}

\end{section}

\begin{section}{Numerical Simulations}

\begin{figure*}[t!]
    \includegraphics[width=\textwidth]{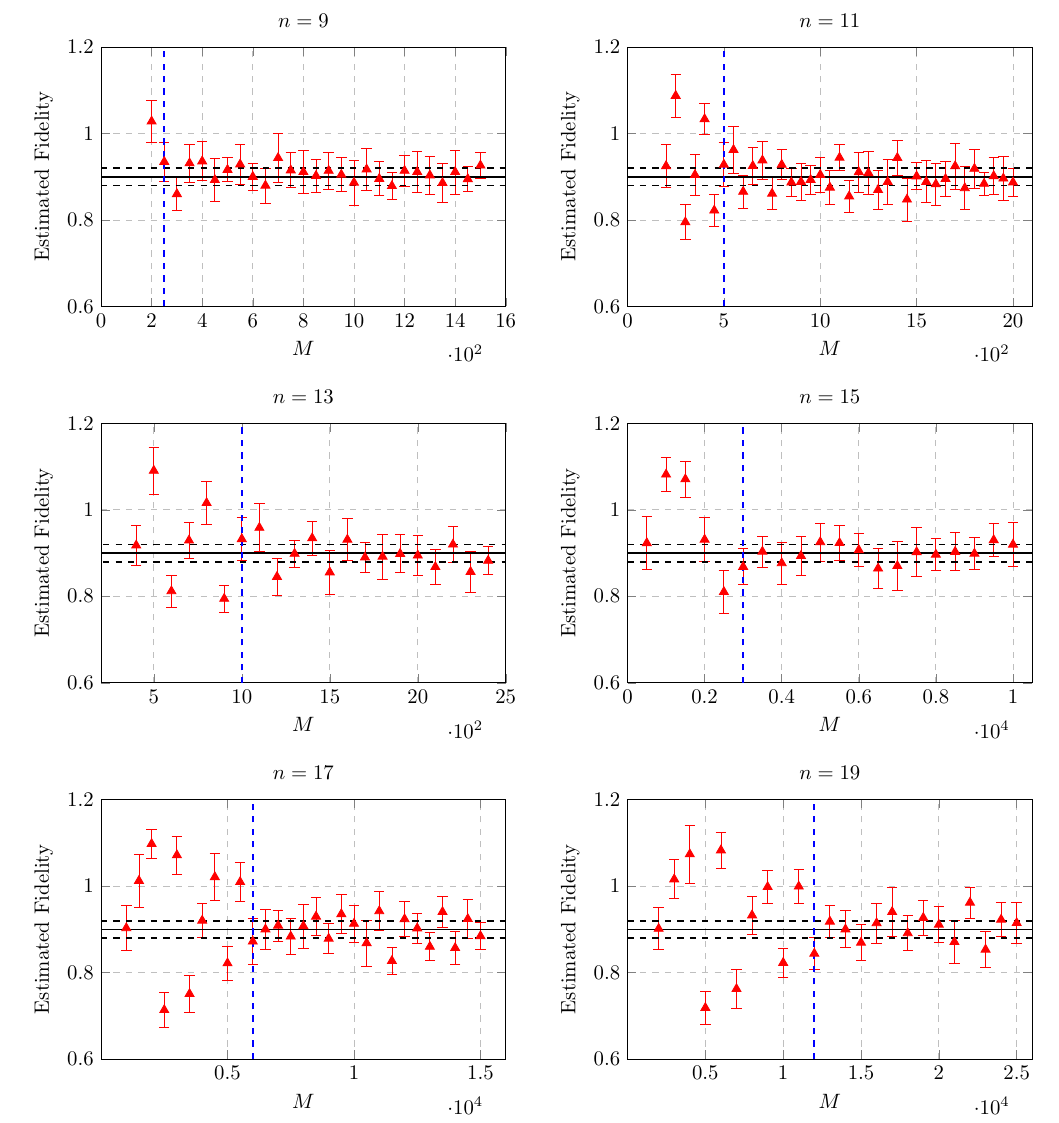}
    \caption{\label{fig:sm:simulations}The mean and standard deviations of the estimated fidelity for various $M$ at different system sizes. The vertical dashed lines correspond to the value of $M$ for which the fidelity estimation achieves the desired precision, as indicated by the horizontal dashed line (see main text for the procedure).}
\end{figure*}

The probability distribution of amplitude estimators for the quantum amplitude estimation protocol is essentially a distribution over the possible bit strings of length $m$. For a bit string $z\in \{0,1\}^m$, let $B(z) \equiv \sum_{l=1}^m z_l 2^{l-1}$, where $z_l$ is the $l$-th bit of $z$. Let $\theta_a \in [0,\pi)$ be such that $|\langle b|C|\psi\rangle|^2 = \sin^2 \theta_a$. For real $\omega_1, \omega_2$, let $d(\omega_1, \omega_2) \equiv \min_{p \in \mathbb{Z}} |p + \omega_1 - \omega_2|$. Recall also that $M \equiv 2^m$. The probability of obtaining the outcome corresponding to a bit string $z\in \{0,1\}^m$ is then~\cite{Brassard2002}
\begin{equation}
    \Pr(\hat{z} = z) = \frac{1}{2} \frac{\sin^2\left(M \pi d(\theta_a/\pi, B(z)/M)\right)}{M^2 \sin^2\left(\pi d(\theta_a/\pi, B(z)/M)\right)}
    + \frac{1}{2} \frac{\sin^2\left(M \pi d(1 - \theta_a/\pi, B(z)/M)\right)}{M^2 \sin^2\left(\pi d(1 - \theta_a/\pi, B(z)/M)\right)}.
\end{equation}
This probability distribution was utilized extensively in our numerical simulations. Details for the simulation of the full fidelity estimation protocol are discussed in the main text.

Another essential component of our numerics was our approach to showing the scaling of the resources required to achieve an accuracy of $0.02$ in our fidelity estimate. For this purpose, we computed empirical distributions of fidelity estimators for varying values of $M$, whose means and standard deviations are shown in Fig.~\ref{fig:sm:simulations}. The horizontal dashed black lines indicate an interval of radius 0.02 around the true fidelity, while the vertical dashed blue lines indicate the value of $M$ which we deemed sufficient to estimate $F$ with accuracy 0.02 (that is, for the majority of higher values of $M$, their error bars lie with in the strip of radius $0.02$ about the true fidelity). These values of $M$ are plotted in Fig.~2 of the main text.

\end{section}

\end{document}